\documentclass{siamart1116}

\numberwithin{equation}{section}
\numberwithin{table}{section}
\numberwithin{figure}{section}

\usepackage{lipsum}
\usepackage{amsfonts}
\usepackage{graphicx}
\usepackage{epstopdf}
\usepackage{algorithmic}
\usepackage[framed,numbered,autolinebreaks,useliterate]{mcode}
\usepackage{amsmath}
\usepackage{amssymb}
\usepackage{algorithm}
\usepackage{bbm}
\usepackage{color}
\usepackage{tikz}
\usetikzlibrary{shapes,arrows}
\newcommand*\diff{\mathop{}\!\mathrm{d}}

\ifpdf
  \DeclareGraphicsExtensions{.eps,.pdf,.png,.jpg}
\else
  \DeclareGraphicsExtensions{.eps}
\fi

\numberwithin{theorem}{section}

\newcommand{\TheTitle}{Monte Carlo pathwise Sensitivities for Barrier Options} 
\newcommand{\TheAuthors}{Thomas Gerstner, Bastian Harrach, Daniel Roth}

\headers{\TheTitle}{\TheAuthors}

\title{{\TheTitle}\thanks{Submitted to the editors 04/18/19.
}}

\author{
  Thomas Gerstner\thanks{Department of Mathematics, Goethe University Frankfurt, Germany (\email{gerstner@math.uni-frankfurt.de},
    \email{harrach@math.uni-frankfurt.de}, \email{roth@math.uni-frankfurt.de}).}
  \and
  Bastian Harrach\footnotemark[2]
  \and
  Daniel Roth\footnotemark[2]
}

\usepackage{amsopn}

\tikzstyle{decision} = [diamond, draw, fill=blue!20, 
    text width=4.5em, text badly centered, node distance=3cm, inner sep=0pt]
\tikzstyle{block} = [rectangle, draw,  
    text width=6em, text centered, rounded corners, minimum height=2.5em]
\tikzstyle{blocks} = [rectangle, draw,  
    text width=5em, text centered, rounded corners, minimum height=2em]
    \tikzstyle{blockss} = [rectangle, draw,  
    text width=2.9em, text centered, rounded corners, minimum height=2em]
\tikzstyle{line} = [draw, -latex']
\tikzstyle{cloud} = [draw, ellipse,fill=red!20, node distance=3cm,
    minimum height=2em]

\ifpdf
\hypersetup{
  pdftitle={\TheTitle},
  pdfauthor={\TheAuthors}
}
\fi

\begin{document}

\maketitle

\begin{abstract}
The Monte Carlo pathwise sensitivities approach is well established for smooth payoff functions. In this work, we present a new Monte Carlo algorithm that is able to calculate the pathwise sensitivities for discontinuous payoff functions. Our main tool is to combine the one-step survival idea of Glasserman and Staum \cite{staum} with the stable differentiation approach of Alm, Harrach, Harrach and Keller \cite{alm}. As an application we use the derived results for a five dimensional calibration of a CoCo-Bond, which we model with different types of discretely monitored barrier options, with time-dependent barrier levels. 
\end{abstract}



\section{Introduction}
Consider Monte Carlo algorithms (see, e.g., Glasserman \cite{glasserman}) for the pricing and the computation of sensitivities for different types of options with discontinuous payoff, specially discretely monitored barrier options. Depending on whether an underlying exceeds a predefined barrier, the payoff of a barrier option may be zero. For this kind of options, there are two substantial types: those which pay zero when there was no barrier crossing, so called 'knock-in' options, and those which pay zero when the barrier was crossed - the 'knock-out' options. It is obvious that barrier options are cheaper than the standard option without a barrier, since they are worthless in more circumstances. For an overview over other exotic options, particularly with discontinuous payoff, we refer to e.g.\ Zhang \cite{zhang}. Many models and algorithms assume continuous monitoring for barrier options, mainly because this leads to analytical solutions. In practice however, many barrier options traded are discretely monitored, not only since practical implementation issues, but also there are some legal and financial reasons, see e.g.\ \cite{kou}.

The price of an option is evaluated by the integral of its expected discounted payoff under a risk-neutral probability measure. For barrier options however, the payoff is discontinuous over the space of all paths. If we look at simple cases, there are analytical formulas for the option price. But if we want to use a more complex stochastic process or a high dimensional model, there won't be useful formulas. As a result of this, it is often useful to use Monte Carlo simulations, which are easy adapted to these models. However, for Monte Carlo algorithms the discontinuous payoff leads to the problem, that the option's sensitivities such as Delta and Vega can't be stably determined from the numerically calculated prices of the standard Monte Carlo algorithm, since even the smallest numerical errors in the price may have arbitrarily large effects on the sensitivities, see e.g.\ \cite{alm, koster}.

Within this work, we derive a Monte Carlo algorithm that allows to calculate the pathwise sensitivities of knock-out barrier options and additionally digital knock-in and knock-out barrier options. The main part of this paper is based on Glasserman and Staum's \cite{staum} one-step survival strategy and the results of Alm et al. \cite{alm}, of which we know that with the approach we can stably determine the option's sensitivities such as Delta and Vega by simple finite differences. The basic idea of Glasserman and Staum \cite{staum} is to use a truncated normal distribution, which excludes the values above the barrier (e.g.\ for knock-up-out options), instead of sampling from the full normal distribution. This approach avoids the discontinuity generated by any Monte Carlo path crossing the barrier, which yields to a Lipschitz-continuous payoff function \cite{alm}. Furthermore, the output allows stable numerical differentiation and leads to a variance reduction.

The new part will be to develop an extended algorithm that estimates the sensitivities for the one-step survival technique directly, without the need of simulation at multiple parameter values as in finite difference. This is an advantage, since the choice for the step-width of the finite difference varies with the input parameters to balance stability and accuracy. 

The sensitivity computation of the approach involves ideas of pathwise sensitivities, symbolic differentiation and automatic differentiation. We refer to \cite{naumann2012art,griewank2008evaluating} for recent work on automatic differentiation (AD) and to \cite{naumann2018adjoint,giles2006smoking,capriotti2010fast} for its adjoint (AAD) mode in computational finance. For several financial applications it could be very helpful to add additional information or structures to AAD tools, e.g. \cite{naumann2018adjoint} shows how an external function interface can reduce memory requirement for common numerical patterns appearing in financial codes. On the opposite, while hand coding can guarantee maximum performance (e.g. for lookback options one just needs to store the maximum or minimum, or one can store results of very expensive operations, e.g. typically {\tt exp}), it could be unfeasible across a large code base. The algorithm uses the one-step survival smoothing technique together with symbolic differentiation to apply an easy understandable and efficient AAD like method, which can be easily modified by hand.

As a final example, we want to calibrate contingent convertibles as an application of the developed theory and to illustrate another benefit of pathwise sensitivities. These are debt instruments, which convert debt into equity upon a trigger event. Contingent convertibles made their entry in the financial world in December 2009. Llodys banking group offered them, by giving their holders the possibility to swap their bonds into this new one. In early February 2011, Credit Suisse managed to attract \$ 2 bn in new capital by this new asset class.
In \cite{spiegeleer} an in-depth analysis of pricing and structuring of these CoCos is given. Spiegeleer and Schoutens show, that a CoCo-Bond can be priced by a Corporate Bond, a knock-in forward and several binary down-in options. We will compare the usage of the standard Monte Carlo, the one-step survival with numerical differentiation and the one-step survival with pathwise sensitivities for calibrating model parameters of a CoCo-Bond.

For a general overview of Monte Carlo Greek computation for all types of options we refer to \cite{giles2006smoking,glasserman,seydel2006tools,capriotti2010fast,burgos2011computation,gilesvibrato}.
There are several ways to overcome the challenges of different exotic options with non-Lipschitz payoff functions investigated in \cite{burgos2011computation}. Especially for barrier options we have to handle discontinuous path-dependent payoff functions, which can be distinguished in the continuously and the discretely monitored case. For Greeks of continuously monitored barrier options we refer to \cite{burgos}, whereas these are handled for general stochastic differential equations with the Multilevel Monte Carlo approach first introduced by Giles \cite{giles2008multilevel,giles2008improved}. This approach uses pathwise sensitivities, but it is not applicable for discretely observed barrier options. By now, the challenges of discretely observed barrier options are handled in the following two ways: payoff-smoothing combined with finite differences or the Likelihood Ratio Method. However, Alm et al. \cite{alm} determined, that the first approach is computationally more efficient. In our new approach we will combine payoff smoothing with pathwise sensitivities to preserve efficiency and to eliminate the need for finite differences.

The structure of this work is as follows. In section 2 we derive our Monte Carlo pricing algorithms for the above-mentioned barrier options and their pathwise sensitivities. Then, we study our algorithms properties and compare the results of the pathwise sensitivities with the finite difference approach in section 3. Furthermore, we will present the case study of a CoCo-Bond in this section. Section 4 contains some concluding marks. 

\section{Monte Carlo one-step survival pathwise sensitivities for discretely observed barrier options}
\label{sec:main}

In this section we will derive our Monte Carlo pathwise sensitivities algorithm for different types of discretely observed barrier options. For this, we will stick to Alm et al. \cite{alm} and use a slightly modified construction of their first algorithm.

We focus on the case, where the options depend on only one underlying asset. Furthermore, we focus on call options, but the conversion to put options is straightforward.\\

Let $S_t$ describe the evolution of the underlying spot price and let $(S_1,\dots,S_T)$ be the vector containing the evaluations at fixed, chronologically sorted, observation dates $(t_1,\dots,t_T)$. We will focus on the Black-Scholes model, where $S_t$ is assumed to follow a geometric Brownian motion
\begin{align*}
\diff S_t&=\mu S_t \diff t + \sigma S_t \diff W_t,
\end{align*}
with $\mu:=r-b$, where $r$ is the risk-free interest rate and $b$ the dividend yield. $\sigma >0$ is the volatility and $W_t$ the standard Brownian motion. 
This model yields to
\begin{align}
S_{j+1}&=S_j \exp \left( \left( \mu - \frac{\sigma^2}{2}\right) \Delta t + \sigma \sqrt{\Delta t} Z_j\right),\label{S}
\end{align}
with $j=0,\dots,T$ and $Z_j$ independent standard normally distributed, with $t_0$ and $S_0=s_0$ the current time and underlying price and step width $\Delta t$, see e.g.\ Hull \cite{hull}.

\subsection{One-step survival pathwise sensitivities for knock-up-out options}

While in this subsection we study the knock-up-out case in depth, we will present other cases less extensive in the following sections.

The payoff of a discretely observed knock-up-out barrier call option is given by
\begin{align}
V(S_1,\dots, S_T)=
   \begin{cases}
     \left(S_T-K\right)^+=:q(S_T) &  \text{if}\max\limits_{j=1,\dots,T}S_j \le B \label{payoff} \\
     0 & \text{otherwise, }  
   \end{cases}
\end{align}
with barrier value $B$, strike price $K$ and observations $j=1,\dots,T$ at the observation dates $(t_1,\dots,t_T)$. 
\begin{definition}\label{thm:PV}
The present value of an option with payoff \eqref{payoff} is given by the  discounted expected payoff
\begin{align*}
PV_{t_0}=e^{-r(t_T-t_0)}\mathbb{E}(V(S_1,\dots,S_T)),
\end{align*}
at the current time $t_0$ and at the time of the final observation $t_T$.
\end{definition}

Starting with the current underlying price $s_0=S_0$ and using \eqref{S} we can generate a path from $s_1$ to $s_T$, by sampling with independent and identically standard normal distributed random variables $Z_{j}\sim \mathcal{N}(0,1)$.
By sampling a sequence of possible realizations $(s_{1,n},\dots ,s_{T,n})$, $n=1,\dots ,N$, of the random variables $(S_1,\dots,S_T)$, we obtain an unbiased standard Monte Carlo estimator for $PV_{t_0}$, see e.g. \cite{glasserman}.

Now we will derive an alternative unbiased Monte Carlo estimator, based on the idea of one-step survival, which allows for pathwise sensitivities. The idea of Alm et al. \cite{alm} using the one-step survival technique of Glasserman and Staum \cite{staum} to obtain stable differentiability can be interpreted in different ways: forcing the path to stay below the barrier or considering an integral splitting, see e.g. \cite{alm}, whilst the latter will be the foundation of our further studies.

For the expectation of a payoff function $PV_{t_0}$, we have
\begin{align*}
PV_{t_0}(S_0)&= e^{-r\Delta t}\int_\mathbb{R}\phi (z)PV_{t_1}(S_1(z))\diff z,
\end{align*}
with the standard normal distribution $\phi$, the time increment between the first observation and the current time $\Delta t:=(t_{1}-t_0)$ and $S_1(z)$ the value of $S$ at the first observation with \eqref{S}. 
In the following work we will assume an equidistant time increment $\Delta t=(t_1-t_0)=\dots=(t_T-t_{T-1})$ between the monitoring dates, leading to $(t_T-t_0)=T \Delta t$, whereas a generalization would be straightforward.
By splitting the integral at the first observation, see e.g.\ \cite{alm}, we obtain
\begin{align}
PV_{t_0}(S_0)&=
e^{-r\Delta t}\left(0 +\int_{S_1(z)<B}\phi (z)PV_{t_1}(S_1(z))\diff z\right)\label{pvt1},
\end{align}
since the payoff will be zero for $S_1(z)\ge B$. One obtains similar formulas for latter steps and as explained in \cite{alm} the integral has to be normalized in every step to ensure a probability density. In this case this leads to:
\begin{align}
PV_{t_0}(S_0)=& e^{-r\Delta t}p_0 \int_{S_1(z)<B}\frac{\phi (z)}{p_0}PV_{t_1}(S_1(z))\diff z.\label{substitution}
\end{align}
Reaching the time of maturity the present value $PV_{t_T}$ simplifies itself to $q(S_{T})$. Note that in practice no $p_t$ will become zero. This method can be interpreted as a special case of importance sampling, see e.g.\ \cite{glasserman}.
Now, different to \cite{alm} or \cite{staum}, we will do an integral substitution to get an easier access to the pathwise sensitivities.
Rewriting the domain of integration, see e.g. \cite{alm} for a solution of $S_1(z)<B$, and using the substitution

\begin{align}
z=\Phi^{-1}(u\cdot p_0) \label{sub1}
\end{align}
we obtain an independent and constant domain of integration for the present value resulting in:
\begin{align}
PV_{t_0}(S_0)
=&  e^{-r\Delta t}p_0\int\limits_{0}^{1}PV_{t_1}(S_1(u))\diff u,\label{pvt11}
\end{align}
with a slight abuse of notation, namely with
\begin{align}
S_{1}(u)&=S_{0}\exp\left( \left( \mu - \frac{\sigma^2}{2}\right) \Delta t + \sigma \sqrt{\Delta t} z(u)\right),\label{z1}\\
z(u)&=\Phi^{-1}(p_0 u).\label{z2}
\end{align}

By iteratively splitting and substituting till the time of maturity, we obtain the following result.
\begin{theorem}\label{thm:OSSpayoff}
The present value of a knock-up-out barrier option with payoff \eqref{payoff} is given by 
\begin{align}
PV_{t_0}(S_0)= e^{-r(t_T-t_0)}\int\limits_{0}^{1}\cdots \int\limits_{0}^{1}p_0\cdots p_T \cdot q(S_T(u^{(T)},\dots,u^{(1)}))\diff u^{(T)}\cdots \diff u^{(1)}\label{payoff2}
\end{align}
with
\begin{align}
p_t&=\Phi\left( \frac{\log (B/S_t(u^{(t)}))-(\mu-\frac{\sigma^2}{2}\Delta t)}{\sigma \sqrt{\Delta t}}\right),\label{pup}\\
S_{t+1}(u^{(t+1)})&=S_{t}(u^{(t)})\exp\left( \left( \mu - \frac{\sigma^2}{2}\right) \Delta t + \sigma \sqrt{\Delta t} \Phi^{-1}(p_t u^{(t+1)})\right),\label{sup}
\end{align}
and the recursion base function $S_0(u^{(0)})=S_0$.
\end{theorem}

\begin{proof}
Iteratively splitting and substituting following the above considerations leads to:
\begin{align*}
\begin{split}PV_{t_0}(S_0)
&=e^{-r\Delta t}p_0  \int_{S_1(z^{(1)})<B}\frac{{\phi} (z^{(1)})}{p_0}\cdot \dotsc \\
&\qquad  \dotsc\cdot  e^{-r\Delta t} p_{T-1} \int_{S_T(z^{(T)})<B}\frac{{\phi} (z^{(T)})}{p_{T-1}}q({S}_T(z^{(T)},\dots,z^{(1)}))\diff z^{(T)}\cdots\diff z^{(1)}\notag\end{split}\\
&=  e^{-r\Delta t}p_0\int\limits_{0}^{1}\cdots e^{-r\Delta t}p_{T-1}\int\limits_{0}^{1}q(S_T(u^{(T)},\dots,u^{(1)}))\diff u^{(T)}\cdots \diff u^{(1)}\\
\end{align*}
\end{proof}
Using \eqref{pup} and \eqref{sup} a Monte Carlo estimator can create the path from $s_1$ to $s_T$, by sampling with independent and identically distributed random variables $u^{(t+1)}\sim \mathcal{U}(0,1)$.
\begin{corollary}\label{thm:OSSMC}
The one-step survival Monte Carlo estimator for the present value of a knock-up-out barrier option given by the average discounted one-step survival payoff
\begin{align*}
\widehat{\text{PV}_N}=e^{-r(t_T-t_0)}\frac{1}{N}\sum\limits_{n=1}^{N} p_{0,n}\cdot \dotsc\cdot  p_{T-1,n}q(s_{T,n})
\end{align*}
is unbiased.
\end{corollary}


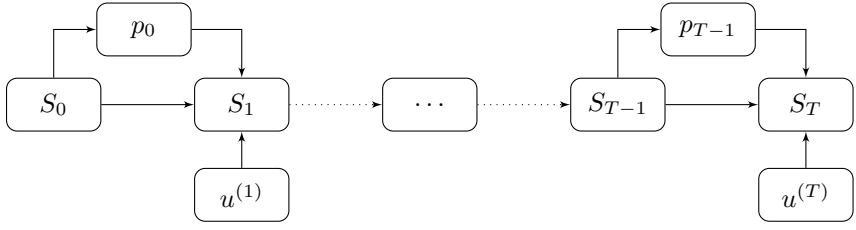
\begin{figure}
\centering
\begin{tikzpicture}[node distance = 2cm, auto]
    \node [blockss] (dritt) {$\hdots$};
    \node [blockss, left of=dritt, node distance=2.5cm] (erst) {$S_{1}$};
    \node [blockss, right of=dritt, node distance=2.5cm] (fuenft) {$S_{T-1}$};
    \node [blockss, right of=fuenft, node distance=2.5cm] (sechst) {$S_{T}$};
    \node [blockss, left of=erst, node distance=2.5cm] (vorher) {$S_{0}$};
  
    \node [blockss, below of=sechst, node distance=1.2cm] (siebt) {$u^{(T)}$};
     \node [blockss, below of=erst, node distance=1.2cm] (vorherunten) {$u^{(1)}$};
       \node [blockss, above of=erst,xshift=-1.3cm, node distance=1.0cm] (oben1) {$p_0$};
          \node [blockss, above of=fuenft,xshift=-1.3cm,xshift=2.5cm, node distance=1cm] (oben2) {$p_{T-1}$};
    \path [line] (erst) -- (dritt)[dotted];
    \path [line] (dritt) -- (fuenft)[dotted];
   \path [line] (vorher) -- (erst);
    \path [line] (vorherunten) -- (erst);
    \path [line] (siebt) -- (sechst);
    \path [line] (fuenft) -- (sechst);
    \path [line] (vorher) |- (oben1);
    \path [line] (oben1) -|(erst);
    \path [line] (fuenft) |-(oben2);
    \path [line] (oben2) -|(sechst);

\end{tikzpicture}
\caption{rescursion formula for $S$ and $p$} \label{fig:recursiontest}
\end{figure}

We derived a recursion formula, illustrated in \cref{fig:recursiontest}, for the modified asset price process and their barrier hitting probabilities.

We see that the integral domains now, as stated above, are compact and the integrand is a composition of Lipschitz-continuous functions. Therefore, the differentiation could be drawn into the integral. In exchange for this, the new asset price $S_{t+1}$ has an extra term depending on $p_{t}$ and its dependencies.

At this point we now want to study the new pathwise sensitivities, for which we will use the derived formulas of \cref{thm:OSSpayoff}. 
In the following we will use $\Theta$ as the variable of differentiation. Furthermore, we will rewrite the recursion and base functions in a more general notation for an easier study of how the derivatives of the paths can be calculated recursively. Let $\Theta:=(\Theta_1,\dots,\Theta_5)=(S_0,B,\mu,\sigma,\Delta t)$, thus \eqref{pup} and \eqref{sup} can be written as
\begin{align}
p_t(\Theta,u)&=f(s,\vartheta)_{\left| \substack{s=S_t(\Theta,u), \vartheta=\Theta}\right.}\label{pup2}\\
S_{t+1}(\Theta,u)&=g(\pi,s,\vartheta,\omega)_{\left| \substack{\pi=p_t(\Theta,u), s=S_{t}(\Theta,u),\omega=u^{(t)},\vartheta=\Theta}\right.}\label{sup2},
\end{align}
with $u=(u^{(1)},\dots,u^{(T)})$ and the recursion base function $S_0(\Theta,u)=\Theta_1$.

For \eqref{pup} and \eqref{sup} this leads to 
\begin{align}
f(s,\vartheta)&=\Phi\left( \frac{\log (\vartheta_2/s)-(\vartheta_3-\frac{\vartheta_4^2}{2}\vartheta_5)}{\vartheta_4 \sqrt{\vartheta_5}}\right),\label{f}\\
g\left(\pi,s,\vartheta,\omega\right)&=s\cdot \exp\left( \left( \vartheta_3 - \frac{\vartheta_4^2}{2}\right) \vartheta_5 + \vartheta_4 \sqrt{\vartheta_5} \Phi^{-1}\left(\pi  \omega\right)\right).\label{g}
\end{align}
Using this notation the present value of the option \eqref{payoff2} can be written as
\begin{align}
PV_{t_0}(\Theta)=  \int\limits_{0}^{1}\cdots \int\limits_{0}^{1}e^{-rT\Delta t} q^*\left(\Theta,u\right)\diff u^{(T)}\cdots \diff u^{(1)},\label{payoffoss}
\end{align}
whereas $q^*\left(\Theta,u\right)$ is the one-step survival payoff defined by
\begin{align}
q^{*}\left(\Theta,u\right):=p_0(\Theta,u)\cdot \dotsc\cdot  p_{T-1}(\Theta,u) \cdot q\left(S_T(\Theta,u)\right).\label{qoss}
\end{align}

Using this notation we can formulate the following result.

\begin{theorem}\label{thm:OSSpathwisesensitivities}
The partial derivatives of the present value of a knock-up-out barrier option with payoff \eqref{payoff} with respect to $\Theta_1,\dots,\Theta_4$ are given by 
\begin{align}
\frac{\partial PV_{t_0}}{\partial \Theta_i}(\Theta)= e^{-rT\Delta t}\int\limits_{0}^{1}\cdots \int\limits_{0}^{1}\frac{ \partial \left(q^{*}\left(\Theta,u\right)\right)}{\partial \Theta_i}\diff u^{(T)}\cdots \diff u^{(1)}\label{pathwisepayoff}
\end{align}
whereas $\frac{\partial q^*}{\partial \Theta_i}(\Theta,u)$ are the derivatives of the one-step survival payoff \eqref{qoss} given by
\begin{align}
\begin{split}\frac{\partial q^*}{\partial \Theta_i}(\Theta,u)=& \mathbbm{1}_{S_T(\Theta,u)>K}\frac{\partial S_T}{\partial \Theta_i}(\Theta,u)\cdot \prod\limits_{j=0}^{T-1}p_j(\Theta,u)\\
 & \hspace{10mm}+q(S_T(\Theta,u))\cdot \sum\limits_{j=0}^{T-1}\frac{\partial p_j}{\partial \Theta_i}(\Theta,u) \prod\limits_{k \neq j}^{T-1} p_k(\Theta,u)\end{split}.\label{pathwiseq}
\end{align}
The derivatives of $p_t(\Theta,u)$ and $S_{T}(\Theta,u)$ are recursively given by
\begin{align}
\frac{\partial p_t}{\partial \Theta_i}(\Theta,u)&=\frac{\partial f}{\partial s}(s,\vartheta)_{\left| \substack{s=S_t(\Theta,u), \vartheta=\Theta}\right.}\frac{\partial S_t}{\partial \Theta_i}(\Theta,u)+\frac{\partial f}{\partial \vartheta_i}(s,\vartheta)_{\left| \substack{s=S_t(\Theta,u), \vartheta=\Theta}\right.}\label{diffp}\\
\frac{\partial S_{t+1}}{\partial \Theta_i}(\Theta,u)&=\frac{\partial g}{\partial s}(\pi,s,\vartheta,\omega)_{\left| \substack{\pi=p_t(\Theta,u), s=S_{t}(\Theta,u),\omega=u^{(t)},\vartheta=\Theta}\right.}\frac{\partial S_t}{\partial \Theta_i}(\Theta,u)\label{diffS}\\
&+\frac{\partial g}{\partial \pi}(\pi,s,\vartheta,\omega)_{\left| \substack{\pi=p_t(\Theta,u), s=S_{t}(\Theta,u),\omega=u^{(t)},\vartheta=\Theta}\right.}\frac{\partial p_t}{\partial \Theta_i}(\Theta,u)\notag\\
&+\frac{\partial g}{\partial \vartheta_i}(\pi,s,\vartheta,\omega)_{\left| \substack{\pi=p_t(\Theta,u), s=S_{t}(\Theta,u),\omega=u^{(t)},\vartheta=\Theta}\right.},\notag
\end{align}
whereas $\frac{\partial S_t}{\partial \Theta_i}(\Theta,u)$ is the derivative of the previous recursion step, with \eqref{pup2}, \eqref{sup2}
and the initial-derivatives:
\begin{align}
\frac{\partial S_0}{\partial \Theta_1}(\Theta,u)=1, \hspace{5mm} \frac{\partial S_0}{\partial \Theta_i}(\Theta,u)=0 \hspace{3mm} \forall i>1. \label{diffbasevector}
\end{align}
\end{theorem}

\begin{proof}
\eqref{pathwisepayoff} follows through
\begin{align*}
\frac{\partial PV_{t_0}}{\partial \Theta_i}(\Theta)&=\frac{\partial \left( \int\limits_{0}^{1}\cdots \int\limits_{0}^{1}e^{-rT\Delta t}q^{*}\left(\Theta,u\right)\diff u^{(T)}\cdots \diff u^{(1)}\right)}{\partial \Theta_i}\notag\\
&= \int\limits_{0}^{1}\cdots \int\limits_{0}^{1}\frac{ \partial \left(e^{-rT\Delta t}q^{*}\left(\Theta,u\right)\right)}{\partial \Theta_i}\diff u^{(T)}\cdots \diff u^{(1)},
\end{align*}
since having compact domains and a Lipschitz-continuous integrand and since the derivatives are not with respect to $r,T$ or $\Delta t$. \eqref{pathwiseq} to \eqref{diffS} are calculated through product rule of \eqref{qoss}, \eqref{pup2} and \eqref{sup2}.
\end{proof}

\begin{figure}
\centering
\begin{tikzpicture}[node distance = 2cm, auto]
    \node [blockss] (dritt) {$\hdots$};
    \node [blockss, left of=dritt, node distance=2.5cm] (erst) {$\frac{S_{1}}{\partial \Theta_i}$};
    \node [blockss, right of=dritt, node distance=2.5cm] (fuenft) {$\frac{S_{T-1}}{\partial \Theta_i}$};
    \node [blockss, right of=fuenft, node distance=2.5cm] (sechst) {$\frac{S_{T}}{\partial \Theta_i}$};
    \node [blockss, left of=erst, node distance=2.5cm] (vorher) {$\frac{S_{0}}{\partial \Theta_i}$};
  
    \node [blockss, below of=sechst, node distance=1.2cm] (siebt) {$u^{(T)}$};
     \node [blockss, below of=erst, node distance=1.2cm] (vorherunten) {$u^{(1)}$};
       \node [blockss, above of=erst,xshift=-1.3cm, node distance=1.0cm] (oben1) {$\frac{p_0}{\partial \Theta_i}$};
          \node [blockss, above of=fuenft,xshift=-1.3cm,xshift=2.5cm, node distance=1cm] (oben2) {$\frac{p_{T-1}}{\partial \Theta_i}$};
    \path [line] (erst) -- (dritt)[dotted];
    \path [line] (dritt) -- (fuenft)[dotted];
   \path [line] (vorher) -- (erst);
    \path [line] (vorherunten) -- (erst);
    \path [line] (siebt) -- (sechst);
    \path [line] (fuenft) -- (sechst);
    \path [line] (vorher) |- (oben1);
    \path [line] (oben1) -|(erst);
    \path [line] (fuenft) |-(oben2);
    \path [line] (oben2) -|(sechst);
\end{tikzpicture}
\caption{Illustration of rescursion dependencies for $\frac{S}{\partial \Theta_i}$ and $\frac{p}{\partial \Theta_i}$.} \label{fig:recursiondiff}
\end{figure}
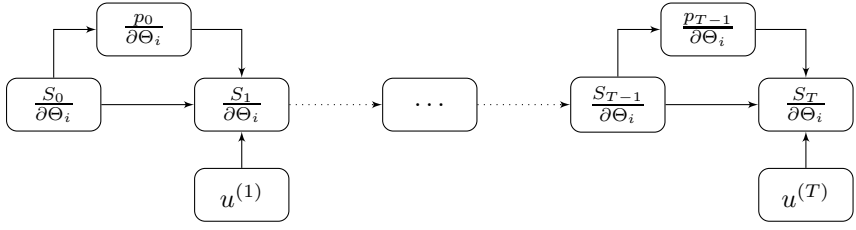

We obtain the expressions \eqref{Df} and \eqref{Dg} for $D^{(\vartheta_1,\dots,\vartheta_4,s)}(f(s,\vartheta))$ and\\ $D^{(\vartheta_1,\dots,\vartheta_4,s,\pi)}(g\left(\pi,u,s,\vartheta\right))$ for our previously introduced model, which can be viewed in the appendix.

\cref{thm:OSSpathwisesensitivities} leads to the following unbiased one-step survival pathwise sensitivities Monte Carlo estimator:

\begin{corollary}\label{thm:MCpathwise}
The one-step survival pathwise sensitivities Monte Carlo estimator with respect to $(\Theta_1,\dots, \Theta_4)$ of the present value of a knock-up-out barrier option given by the average
\begin{align*}
 \begin{split}\widehat{D_{\Theta_i}\text{PV}_N}=&e^{-r(t_T-t_0)}\frac{1}{N}\sum\limits_{n=1}^{N}\left( \mathbbm{1}_{s_{T,n}>K}\frac{\partial s_{T,n}}{\partial \Theta_i}(\Theta,u)\cdot \prod\limits_{j=0}^{T-1}p_{j,n}(\Theta,u)\right.\\
 & \hspace{10mm}\left.+q(s_{T,n}(\Theta,u))\cdot \sum\limits_{j=0}^{T-1}\left[\frac{\partial p_{j,n}}{\partial \Theta_i}(\Theta,u) \prod\limits_{k=0,k \neq j}^{T-1} p_{k,n}(\Theta,u)\right]\right)\end{split}
\end{align*}
is unbiased.
\end{corollary}

We want to remark, that if one is interested in second order Greeks the indicator function at the final step can be smoothed out by using a combination of the methods of this and the next section, forcing the path to stay between $B$ and $K$. 

\begin{corollary}\label{thm:MCpathwisesecond}
The second order one-step survival pathwise sensitivities Monte Carlo estimator with respect to $(\Theta_1,\dots, \Theta_4)$ of the present value of a knock-up-out barrier option is given by:
\begin{align*}
 \begin{split}&\widehat{D_{\Theta_i}\text{PV}_N}
 =e^{-r(t_T-t_0)}\frac{1}{N}\sum\limits_{n=1}^{N}\\
 &\left( \frac{\partial^2 s_{T,n}}{\partial \Theta_i^2}(\Theta,u)\cdot \prod\limits_{j=0}^{T-1}p_{j,n}(\Theta,u)
 +  2  \frac{\partial s_{T,n}}{\partial \Theta_i}(\Theta,u)\cdot \sum \limits_{j=0}^{T-1} \left( \frac{\partial p_{j,n}}{\partial \Theta_i}\prod\limits_{k=0, k\neq j}^{T-1}p_{k,n}(\Theta,u)  \right)  \right.\\
 & \left.+s_{T,n}(\Theta,u)\cdot \sum\limits_{j=0}^{T-1}\left[
 \frac{\partial p_{j,n}}{\partial \Theta_i} \prod\limits_{k=0}^{j-1} p_{k,n}\sum\limits_{k=j+1}^{T-1}\left(\frac{\partial p_{k,n}}{\partial \Theta_i}\prod\limits_{m=j+1,m\neq k}^{T-1} p_{m,n}\right)\right.\right.\\
&\left. \left. + \frac{\partial p_{j,n}}{\partial \Theta_i} \prod\limits_{k=j+1}^{T-1} p_{k,n}\sum\limits_{k=0}^{j-1}\left(\frac{\partial p_{k,n}}{\partial \Theta_i}\prod\limits_{m=0,m\neq k}^{j-1} p_{m,n}\right)  +\frac{\partial^2 p_{j,n}}{\partial \Theta_i^2}(\Theta,u) \prod\limits_{k=0,k \neq j}^{T-1} p_{k,n}(\Theta,u) \right]\right)\end{split}
\end{align*}
\end{corollary}

While in \eqref{Df} and \eqref{Dg} we gave the results of the needed derivatives for the pathwise sensitivities estimator of \cref{thm:MCpathwise}, these can be calculated with an AD like method by e.g.\ an easy MATLAB \cite{MATLAB:2015} script. To explain this idea in more detail, we present \cref{test1}, which uses the Symbolic Math Toolbox\texttrademark          \hspace{1mm}of MATLAB for the differentiation. The algorithm calculates the needed derivatives for \eqref{diffp} and \eqref{diffS} and inserts these into the Monte Carlo simulation, while the derivatives of the payoff are taken out of \cref{thm:MCpathwise}, respectively by hand. 

A straightforward procedure holds for second order Greeks while using similar MATLAB commands for second order differentiation and taking the derivatives of the payoff out of \cref{thm:MCpathwisesecond}. \\

\begin{algorithm}\label{algorithm1}
\caption{One-step survival pathwise sensitivities estimator with respect to $S_0,B,\mu,\sigma$  (Delta, Beta, Rho, Vega) of a knock-up-out barrier option.}\label{test1}
\begin{algorithmic}[1]
\STATE \% \textit{symbolic definition of \eqref{f} and \eqref{g}}
\STATE syms $\vartheta_1,\dots,\vartheta_5$,  $u$, $\pi$, $s$
\STATE$f=\Phi\left((\text{ln}(\vartheta_2/s-(\vartheta_3-\vartheta_4^2/2) \vartheta_5)/(\vartheta_4\sqrt{\vartheta_5})\right)$ \, \% \textit{$f(s,\vartheta)$}
\STATE $g=s\cdot \exp\left((\vartheta_3-\vartheta_4^2/2) \vartheta_5+\vartheta_4 \sqrt{\vartheta_5}\cdot \Phi^{-1}(u\cdot \pi)\right)$ \, \% \textit{$g(\pi,s,\vartheta,\omega)$}
\STATE
\STATE \% \textit{symbolic partial derivatives of $f$ and $g$, resulting in \eqref{Df} to \eqref{Dg} and the conversion of the symbolic expressions to function handles}
\STATE $f$=matlabFunction($f$)
\STATE $g$=matlabFunction($g$)
\STATE ${D_\vartheta}f$=matlabFunction(jacobian($f,[\vartheta_1,\dots,\vartheta_4]$))
\STATE ${D_s}f$=matlabFunction(jacobian($f,s$))
\STATE ${D_\vartheta}g$=matlabFunction(jacobian($g,[\vartheta_1,\dots,\vartheta_4]$))
\STATE ${D_s}g$=matlabFunction(jacobian($g,s$))
\STATE ${D_\pi}g$=matlabFunction(jacobian($g,\pi$))

\STATE
\STATE \% \textit{Monte Carlo simulation}
\STATE Initialize random seed
\STATE Initialize model parameters $\Theta=(\Theta_1,\dots,\Theta_8)=(S_0,B,\mu,\sigma,\Delta t,r,K,T)$
\FOR {$n=1,\dots,N$} 
\STATE \% \textit{base derivative recursion vector as in \eqref{diffbasevector}}
\STATE $D_\Theta{S}_0=[1,0,0,0]$ 
\FOR {$j=0:T-1$}
\STATE \% \textit{simulate paths as in \eqref{pup2} and \eqref{sup2}}
\STATE $p_j:=f(S_j,\Theta)$
\STATE Sample $u \sim U(0,1)$
\STATE $S_{j+1}:=g(p_j,S_j,\Theta,u)$
\STATE \% \textit{simulate derivatives of paths as in \eqref{diffp} and \eqref{diffS}}
\STATE $D_\Theta{p}_{j} :=D_sf(S_j,\Theta)\cdot D_\Theta{S}_j+D_\vartheta f(S_j,\Theta)$
\STATE $D_\Theta{S}_{j+1}:=D_s g(p_j,S_j,\Theta,u)\cdot D_\Theta{S}_j$
\STATE \hspace{30mm}$+D_\pi g(p_j,S_j,\Theta,u)\cdot D_\Theta{p}_j+D_\vartheta g(p_j,S_j,\Theta,u)$
\ENDFOR
\STATE \% \textit{calculate price as in \cref{thm:OSSMC}}
\STATE $P_n:= \text{prod}(p)\cdot \max(S_T-K,0)$
\STATE \% \textit{calculate the pathwise sensitivities as in \cref{thm:MCpathwise} }
\STATE $D_\Theta{P}_{n}:=\mathbbm{1}_{S_T>K}\cdot D_\Theta{S}_T \cdot\text{prod}(p) $
\FOR{$i=1,\dots,T$}
\STATE $D_\Theta{P}_{n}:=D_\Theta{P}_n+ \max(S_T-K,0) \cdot D_\Theta{p}_i \cdot \text{prod}(p)/p_i $
\ENDFOR

\ENDFOR
\STATE \textbf{return} $PV_{t_0}:=e^{-r\cdot T\Delta t}\frac{1}{N}\sum_{n=1}^N P_n$ , $D_\Theta{PV}_{t_0}=e^{-r\cdot T\Delta t}\frac{1}{N}\sum_{n=1}^N D_\Theta{P}_{n}$
\end{algorithmic}
\end{algorithm}

For a better understanding, since \cref{test1} uses the syntax and some functions of MATLAB, we will explain the code line by line in the following paragraph.

In the rows $3$ and $4$ we define the functions \eqref{f} and \eqref{g}, whereas all variables are handled in the symbolic way. To generate the symbolic expressions, we use the MATLAB function {\tt syms}.
From row $7$ to $13$ the algorithm calculates the symbolic derivatives of $f$ with respect to $\vartheta_1,\dots,\vartheta_4$ and $s$ and of $g$ with respect to $\vartheta_1,\dots,\vartheta_4,s$ and $\pi$, with the symbolic MATLAB function {\tt jacobian}. Remember that in this script these are equal to the matrices \eqref{Df} and \eqref{Dg} of the appendix.
After calculating the symbolic derivatives, the algorithm converts the symbolic expressions to function handles with {\tt matlabFunction()}. 

After this predefining of the functions, the algorithm starts with the Monte Carlo simulation at row $17$. In addition to the asset price and survival value simulation in lines $23$ to $25$, the algorithm calculates the derivatives of these in the lines $27$ and $38$ as in the formulas derived in \eqref{diffp} and \eqref{diffS}.

After the simulation of the paths the algorithm first calculates the one-step survival payoff as in \cref{thm:OSSMC} in line $35$ and the pathwise sensitivities of the option (Delta, Beta, Rho, Vega) from line $37$ to $48$ as in \cref{thm:MCpathwise}, with a slight modification emphasizing the relevance of hand coding which we want to explain in more detail: In terms of complexity \cref{thm:MCpathwise} tends to be $\mathcal{O}(NT^2)$, with $N$ Monte Carlo simulations and $T$ observations. As seen in line $36$ we can reduce the complexity to  $\mathcal{O}(NT)$ by replacing $\prod_{k=0,k \neq j}^{T-1} p_{k}$ with $(\prod_{k=0}^{T-1} p_{k})/p_{j}$, since the product can be precalculated.

We want to remark that we didn't have any problems with $p_j$ close to zero. On the one hand there are no cancellation or absorption problems with division, on the other hand if $p_j$ equals zero (very unlikely for Black-Scholes, but we could imagine it on the last step if using a numerical approximation, e.g. Euler-Maryuama, for more general models) the division wouldn't be necessary, since the path would not 'survive'. Nonetheless, alternatively $\prod_{k=0, k\neq j}p_k$ could be computed through multiplication of two pre-calculated cumulative products (increasing respectively decreasing) without increasing the order, nevertheless the pre-allocation increases the memory usage. 
Furthermore, one could use similar ideas to reduce the complexity of second order Greeks of \cref{thm:MCpathwisesecond}, i.e.\ almost $\mathcal{O}(NT^3)$, to an optimal complexity of $\mathcal{O}(NT)$. We will shortly explain some needed modifications:
$\prod_{k=0}^{j-1} p_{k}$ and $\prod_{k=j+1}^{T-1} p_{k}$ are exactly the cumulative products described for Delta. The remaining expressions can be modified through recursion formulas of the form:
\begin{align*}
 \sum\limits_{k=0}^{j}\left(\frac{\partial p_{k}}{\partial \Theta_i}\prod\limits_{m=0,m\neq k}^{j} p_{m}\right)=p_{j}\sum\limits_{k=0}^{j-1}\left(\frac{\partial p_{k}}{\partial \Theta_i}\prod\limits_{m=0,m\neq k}^{j-1} p_{m}\right)+\frac{\partial p_{j}}{\partial \Theta_i}\prod\limits_{i=0}^{j-1}p_i,
\end{align*}
i.e. starting with $j=1$, we can recursively calculate the expressions, while saving them in a vector. Finally, the vector can be summed up in order of $T$.

Finally, for this section we will give some remarks for the particular implementation in MATLAB. In the used MATLAB version the function matlabFunction() sets the input parameters in alphabetic order, which we therefore considered in the algorithm. 

To implement the algorithm, the normal distribution and the inverse normal distribution should be replaced by the formulas
\begin{align*}
\Phi(x)&=0.5\left(\text{erf}\left(\frac{x}{\sqrt{2}}\right)+1\right)\\
\Phi^{-1}(x)&=\text{inverf}(2x-1)\sqrt{2},
\end{align*}
since MATLAB is not able to differentiate symbolically the \textit{norminv} function at this time.

\subsection{One-step survival for other type of barrier options}

In this section we will shortly explain the changes that need to be made for other type of barrier options. For knock-up-out barrier options we used \eqref{sub1} for the substitution. However, for knock-down-out options defined by
\begin{align}
V(S_1,\dots,S_T)=
   \begin{cases}
     \left(S_T-K\right)^+=:q(S_T) & \text{if}\min\limits_{j=1,\dots,T}S_j \ge B  \\
     0 & \text{otherwise. }  
   \end{cases}\label{payoffdown}
\end{align}
the path survives while staying above the barrier. Thus, after splitting the integral and with a modified normalization, we obtain
\begin{align}
PV_{t_0}(S_0)=&  e^{-r\Delta t}(1-p_0)\int_{S_1(z)\ge B}\frac{\phi (z)}{(1-p_0)}PV_{t_1}(S_1(z))\diff z,\label{knockdownint}
\end{align}
at the first observation date. Now, demanding an independent and compact domain of integration, we use
\begin{align}
z=\Phi^{-1}((1-p_0)u+p_0)\label{sub2}
\end{align}

for the substitution. Now, similar to \cref{thm:OSSMC} and \cref{thm:OSSpayoff}, one could formulate similar results for knock-down-out options. At this point we just present the essential consequences for \cref{thm:OSSpathwisesensitivities}.

First we determine that instead of \eqref{qoss} we obtain the modified one-step survival payoff
\begin{align}
q^{*}\left(\Theta,u\right):=(1-p_0(\Theta,u))\cdot \dotsc\cdot (1-p_{T-1}(\Theta,u)) \cdot q\left(S_T(\Theta,u)\right),\label{downpayoff}
\end{align}
and instead of \eqref{g} we have to use
\begin{align}
g\left(\pi,u,s,\vartheta\right)&=s\cdot \exp\left( \left( \vartheta_3 - \frac{\vartheta_4^2}{2}\right) \vartheta_5 + \vartheta_4 \sqrt{\vartheta_5} \Phi^{-1}\left((1-\pi)  u+\pi\right)\right)\label{newg}
\end{align}

All in all, to obtain the one-step survival Monte Carlo estimator for the pathwise sensitivities and the payoff of a knock-down-out barrier option \cref{test1} only has to be modified at line $4$ with the new $g(\pi,s,\vartheta,\omega)$ of \eqref{newg} and at the lines $31$ to $37$ with the new payoff \eqref{downpayoff} and its derivatives, calculated straightforward by hand, as seen in \eqref{pathwiseq}.
 
The introduced techniques can't be easily applied to knock-in options, since none of the split integrals become zero.

For digital knock-in barrier options defined by:  
\begin{align}
V(S_1,\dots, S_T)=
   \begin{cases}
      c=:q(S_T) & \text{if} \max\limits_{j=1,\dots,T}S_j \ge B \\
     0 & \text{otherwise, }  
   \end{cases}\label{payoffdigital}
\end{align}
we can apply one-step survival similarly, since one part of the split integrals will be constant and doesn't need to be simulated any further.
By splitting, normalizing and substituting the integral at the first observation date, with \eqref{sub1} for the first summand and \eqref{sub2} for the second, we obtain
\begin{align*}
PV_{t_0}(S_0)
=
e^{-r\Delta t}\left((1-p_0)\cdot c + p_0\int\limits_{0}^{1}PV_{t_1}(S_1^{(1)}(u))\diff u\right),
\end{align*}
with \eqref{pup} for $p_0$ and \eqref{sup} for $S_1^{(1)}(u)$ 
and since the payoff will be $c\in \mathbb{R}$ if the underlying hits the barrier. 

All in all we see that the recursion formulas from section 2.1 hold here and a theorem for the present value and the sensitivities can be formulated analogue to \cref{thm:OSSpayoff} and \cref{thm:OSSpathwisesensitivities} while using the modified one-step survival payoff
\begin{align*}
q^*(\Theta,u)=c\left((1-p_0) +p_0  (1-p_1)  +  \dotsc  + p_0 \cdot \dotsc\cdot  p_{T-2} (1-p_{T-1})\right)
\end{align*}
and its straightforward calculated derivatives.

Finally, we remark that the theory and \cref{test1} can be adjusted straightforward for knock-down-in digital options with these results. As mentioned before, it is not straightforward to use the presented techniques for general (non-digital) knock-in barrier options. However, in practice one can use the in-out-parity
\begin{align}
\frac{\partial (V^{\text{knock-in}}(S,T))}{\partial \Theta_i}&=\frac{\partial (V^{\text{opt.}}(S,T))}{\partial \Theta_i} - \frac{\partial (V^{\text{knock-out}}(S,T))}{\partial \Theta_i},\label{inequalsout}
\end{align}
whereas for the differentiation of the non-barrier option we can use e.g.\ the standard pathwise sensitivities approach from \cite{glasserman}.
As a little remark, we want to mention that if the variable $\Theta_i$ is the barrier $B$ the differentiation of the plain option $\frac{\partial (V^{\text{opt.}}(S,T))}{\partial \Theta_i}$ will become zero.

\section{Numerical results}

In this section, we will provide some numerical results for the new introduced algorithm. Therefore, we consider a simple discretely observed up-and-out barrier option in two different settings, once with $50$ and once with $360$ observations before maturity. We will use the parameter values of Table \ref{parameter}, whereas the example is fictitious.

\begin{table}
\centering
  \begin{tabular}{lr}
  \hline
 Parameter & Value\\
 \hline
 $t_0$ & 0 \\
 $t_T$ & 1  \\
 $S_0$ & $50$  \\
 $B$ & $60$ \\
 T (\# observation dates) & $(50,360)$\\
 $r$ & 10 \%\\
 $b$ & 0 \%\\
 $\sigma$ & $20$ \%\\
 $K$ & $50$ \\
 \hline
 \end{tabular}
 \caption{Parameters for both up-and-out barrier option, only differing in the amount of observation dates. The observations are distributed equidistantly till maturity. For that reason the first observation will be at $(t_0-t_T)/T=1/50$ for the first example and respectively at $1/360$ for the second.}
  \label{parameter}
\end{table}

In the first column of \cref{fig:vergleich} we see the estimated value of the option as a function of the initial asset price at $t_0$. The results of the standard Monte Carlo estimator and the one-step survival Monte Carlo estimator are plotted in a black dashed line, and in a gray solid line, respectively. The rows use $N=10^2,N=10^3,N=10^4,$ and $N=10^5$ Monte Carlo samples from top to bottom.
For each of these and the following calculations, the same random seed was used.
The second column shows the first derivative of the option value with respect to the underlying price $S_0$ (the Delta) calculated by applying central finite differences
\begin{align*}
\frac{\widehat{PV}(S_o+\delta_s)-\widehat{\text{PV}}(S_0-\delta_s)}{\delta_s}, 
\end{align*}
with $\delta_s$ chosen as $0.5\%$ of $S_0$.
The third column shows the comparison of the one-step survival finite difference derivatives and the new one-step survival pathwise derivatives.
\begin{figure}
\centering
\includegraphics[scale=.4]{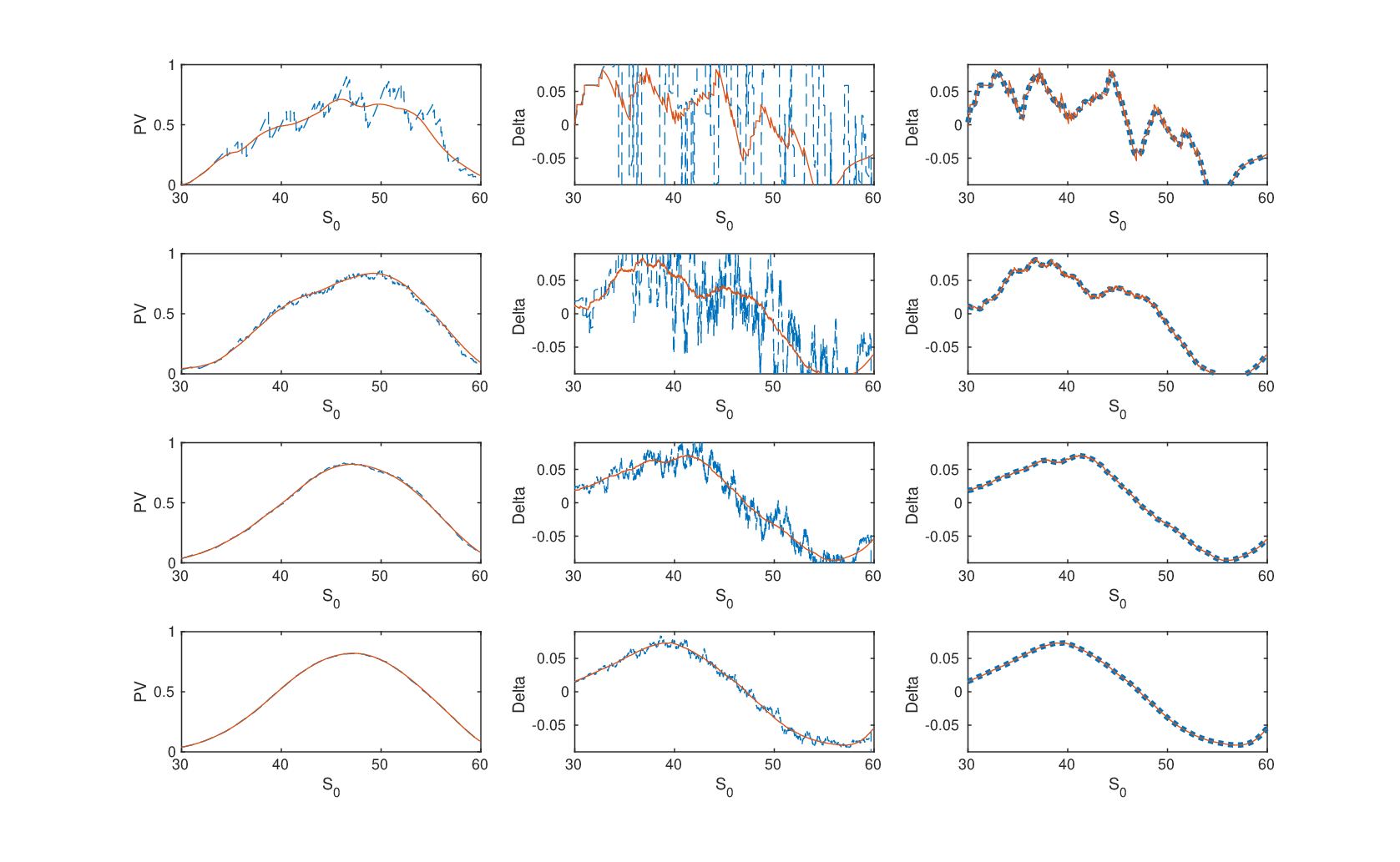}
\vspace{-7mm}
\caption{First column shows the value of the option, the case with $50$ observation dates, calculated with standard Monte Carlo (dashed line) and with the one-step survival estimator (solid line) as a function of the asset price $S_0$ at $t_0$. Second column shows the respective first derivatives by finite difference with standard Monte Carlo (dashed line) and pathwise sensitivities (solid line). The third column shows the new pathwise sensitivity approach (solid line) and again with finite differences (dotted line). In the lines the Monte Carlo samples are increased from $10^2$ to $10^5$.}
\label{fig:vergleich}
\end{figure}
The plots clearly demonstrate the instability of the standard Monte Carlo estimator with respect to numerical differentiation and the stability of the one-step survival MC estimator, as already mentioned in \cite{alm}. Furthermore, we see that the pathwise sensitivities are quite close to the finite differences. Hence, we want to take a deeper look at the comparison of these two methods for sensitivity computation, due to being more complex, as developed in \eqref{diffp} and \eqref{diffS}, the pathwise sensitivity approach doesn't need any additional path for the sensitivity computation. We present the CPU times (computed on a single core of an Intel i7-4790 CPU) of first and second order pathwise sensitivities respectively finite differences (calculated through first and second order difference quotients) for both parameter settings in \cref{cputime}. Furthermore, we compare the absolute errors for Delta and Gamma in \cref{fig:timeerror1}.

\begin{table}
 \small
\centering
  \begin{tabular}{c|ccc}
    \multicolumn{4}{c}{\textbf{computation time (CPU) in seconds}}\\
  \hline

T (\#observations) & present value & Delta: PW (FD) & Delta+Gamma: PW (FD) \\
 \hline
50 &0.2285&0.3289 (0.4571)&0.5565 (0.6857)\\
360 &1.56283&2.2407 (3.1256)&3.7560 (4.6885)\\
 \hline
 \end{tabular}
 \caption{Computation time of the one-step survival algorithm for the present value and for Delta and Gamma (including Delta) with pathwise sensitivities (PW) and central finite differences (FD) for both options computed with fixed Monte Carlo sample size ($10^5$). }
  \label{cputime}
\end{table}

\begin{figure}
\centering
\includegraphics[scale=.35]{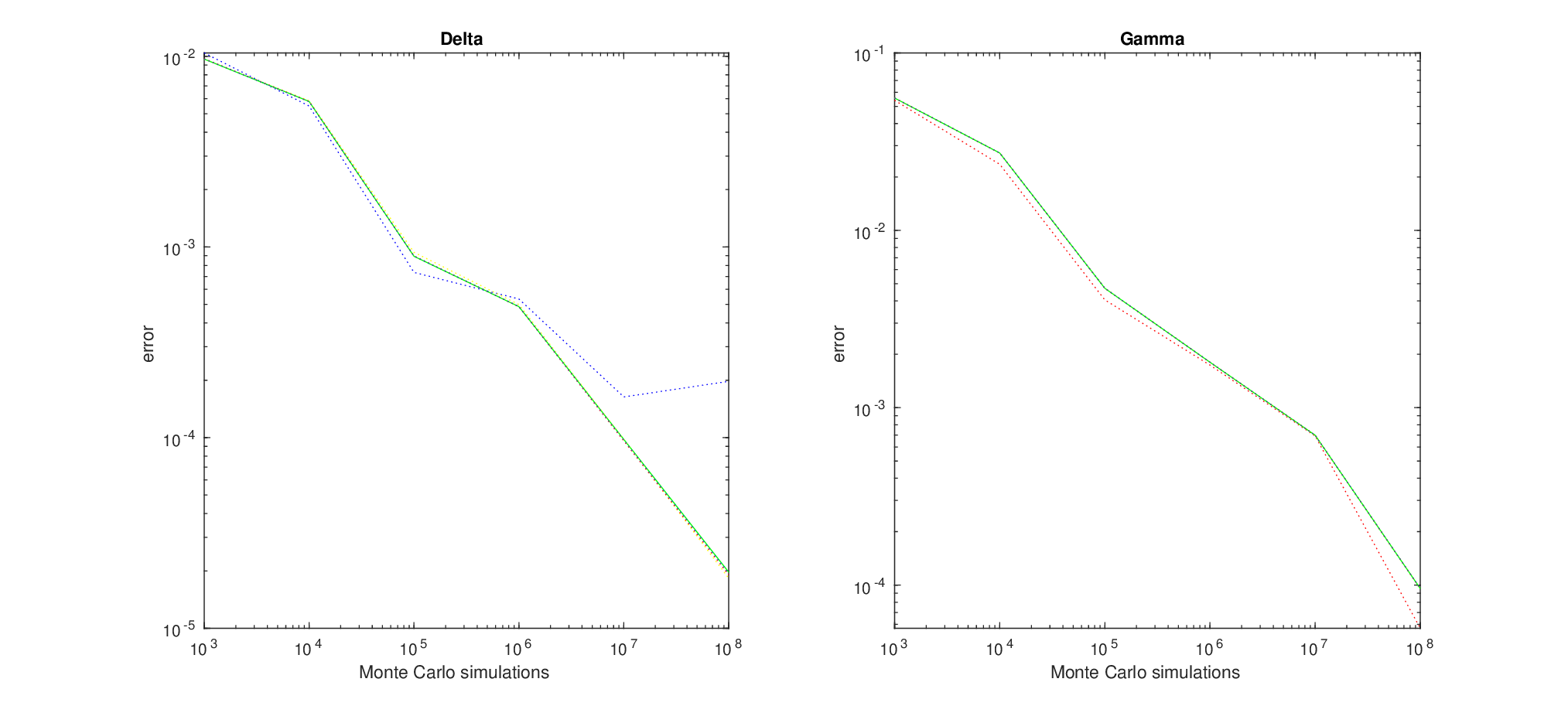}
\caption{Delta and Gamma of the pathwise sensitivity approach (solid line) vs. finite differences approaches, with $\delta_s=10^{(-0)},\delta_s=10^{(-1)},\delta_s=10^{(-2)},\delta_s=10^{(-3)}$ (dashed lines). The Figures show the absolute error of a fixed $S_0$ depending on the number of Monte Carlo simulations.}
\label{fig:timeerror1}
\end{figure}
First, we determine that the made considerations and modifications regarding the complexity with respect to the observations hold.

Next, we see, with regard to the absolute error, that the new pathwise estimator doesn't perform much better (but without discretization error), but, as can be viewed in \cref{cputime}, it needs substantially less time related to finite differences: $\approx 0.71$ ratio for Delta respectively $\approx 0.8$ ratio for Gamma.

\subsection{Calibration of a CoCo-Bond}

In this section, on the one hand, we will take a closer look at the time-saving factor for multiple Greeks and on the other hand, we will determine another benefit of pathwise sensitivities. Since being independent of discretization errors, the new approach provides a larger radius of convergence for calibrations.

CoCo-Bonds are debt instruments converting upon a certain trigger event. There are several common used options for defining this trigger-event: an accounting trigger, a multi-variate trigger or a regulatory trigger, see e.g. \cite{spiegeleer} for explanations. 

Here, for modelling CoCo-Bonds, we will use the equity derivative approach of Spiegeleer and Schoutens \cite{spiegeleer}. As the authors explain this approach does not model the true trigger-event but rather uses the approximate model that the bond is triggered whenever the underlying falls below a certain level $B$, where $B$ is to be calibrated from market data.

Let us stress that this makes calibration problems for CoCo-Bonds fundamentally different to the case of standard barrier options. For the pricing of standard barrier options, one can completely avoid calibration problems with discontinuous payoffs, since the barrier is a known pay-off feature, and other model parameters such a $\mu$ and $\sigma$ are model parameters of the underlying and can thus be calibrated through e.g. European options. However, for Coco-Bonds, using the Spiegeleer-Schoutens-model \cite{spiegeleer}, the barrier $B$ is a model parameter that can only be determined by calibrating the Coco-Bond itself. Thus, Coco-Bonds require stable calibration methods for discontinuous payoffs.

Note also that, while a one dimensional calibration of $B$ is unavoidable, the calibration of $\mu$ and $\sigma$ can be done either from standard European option market prices (for simpler calibration), or also from CoCo-Bond market prices (for higher consistency). The latter leads to multi-dimensional calibration problems for discontinuous barrier options. Also, multi-dimensional calibration seems unavoidable for multivariate underlyings or when the barrier $B$ is modelled to be time-dependent. To demonstrate the feasibility of our approach for multi-dimensional calibration problems while staying comparable with the example in \cite{spiegeleer}, we will treat the example of a univariate CoCo-Bond, while modelling five variable time-depending barrier levels.

Since we use the equity derivative approach of \cite{spiegeleer} the holder obtains coupons as long as the trigger event has not occurred. The trigger event is observed discretely at the dates of coupon payments. While the coupons will end when the trigger event is fact, the value of the CoCo-Bond is reduced compared to a plain corporate bond of the same issuer. This is valued as a short position in a binary down-and-in barrier option. We notice, that for every coupon, there is indeed a corresponding short position in binary option that is knocked in on the same barrier.

The used parameters can be viewed in \cref{table:parametersCoCo}, see \cite{spiegeleer} for further information of payoff structures. For the individual barrier levels, we modelled a convex 'smile'-pattern.

\begin{table}
\centering
  \begin{tabular}{lr}
  \hline
 Parameter & Value\\
 \hline
 $T$ & $4.5$ \\
 $r$ & 3.42 \% \\
 $\sigma$ & 10 \%\\
 $(S_0,S_{0.5},S_1,S_{1.5},S_{2})$ & $(0.6075,0.61,0.6025,0.6125,0.605)$\\
 $C_p$ & 0.5900\\
  $(B_{0.5},\dots,B_{2})$ & $0.5$ \\
   $(B_{2.5},B_{3},B_{3.5},B_{4},B_{4.5})$ & $(0.51,0.48,0.45,0.49,0.52)$ \\
 COUPON & 30\% \\
 FREQUENCY & Semi-Annual \\
 FACE VALUE & 100\\
 CONVERSATION RATIO & 100\\
 \hline
 \end{tabular}
 \caption{Used parameters for the model and the CoCo-Bond price computations starting in ascending semi-annual sequence $(0,0.5,1,1.5,2)$. The first five barrier levels are assumed to be constant.}
 \label{table:parametersCoCo}
\end{table}

We calculated the benchmark prices of the CoCo-Bond with a high number of Monte Carlo samples.

\begin{table}
 \small
\centering
  \begin{tabular}{|ccc|}
    \multicolumn{3}{c}{CPU time of a \textbf{single} iteration}\\
  \hline
 \# simulations  & OSS  & PW \\
 \hline  \vspace{-3.5mm}\\
$10^7$  & $10669.20$& $2196.17$\\
$10^6$  & $1067.34$ & $219.7$\\
$10^5$  & $106.02$ & $22.00$\\
 \hline
 \end{tabular}
 \caption{CPU time (in seconds) of one iteration step (including all derivatives) of the one-step survival estimator (OSS) using finite differences and the one-step survival pathwise sensitivity estimator (PW). }
  \label{cputime2}
\end{table}

The results are demonstrated in the following order: First, we compare the CPU computation time (\cref{cputime2}) of one-step survival using finite differences and the one-step survival pathwise sensitivities for a single calibration iteration step (i.e. present value and all derivatives). Next we will compare the calibrations using a gradient-based algorithm (\cref{table:calibrationOSSPW}), while starting with initial-vector $(0.4,0.4,0.4,0.4,0.4)$. Last, using a slightly modified initial-vector, namely $(0.3,0.3,0.3,0.3,0.3)$, we will determine an increased radius of convergence (\cref{table:calibrationOSSPWradius}), using the new pathwise sensitivity approach. 

For illustration purpose, all MATLAB optimization algorithms are used as black boxes, without any modification and without any additional data. 

Studying the results of {\tt lsqnonlin}, see \cref{table:calibrationOSSPW}, which is based on a trust-region-reflective method we see the expected weakness of standard Monte Carlo. Even for a huge number of samples the standard Monte Carlo estimator is not able to provide viable solutions, since instantly aborting. Hence, for calibrations with standard Monte Carlo estimators, one would typically resort to regularized differentiation schemes.
The calibrations, using the pathwise sensitivity approach, deliver similar residuals, while requiring similar iterations with approximately a fifth of the time (as one-step survival with finite differences, see \cref{cputime2}). Hence, we clearly see the expected strength of the pathwise sensitivities for multidimensional cases. Furthermore, we determine, see \cref{table:calibrationOSSPWradius}, that the new approach provides a larger radius of convergence, while finite differences deliver unfeasible results. 

We remark, that while experimenting with these multi-dimensional calibration, we also detected advantages of one-step survival, using gradient-free algorithms (e.g.  {\tt fminsearch}), compared to standard Monte Carlo. This effect is based on the idea that unstable differentiation also influences gradient free algorithms.

{
\begin{table}
\centering
  \begin{tabular}{|cccc|}
    \multicolumn{1}{c}{{\tt lsqnonlin:}} &   \multicolumn{2}{c}{\textbf{\textbf{standard}}}\\
  \hline
\# MC & iterations & result & resnorm\\
 \hline  \vspace{-3.5mm}\\ 
 $10^7$ & $0$ & $(0.4000,0.4000  ,  0.4000  ,  0.4000    ,0.4000) $ & $152.3039$\\
 $10^6$ & $0$  & $(0.4000,0.4000  ,  0.4000  ,  0.4000    ,0.4000) $ & $152.1053$\\
 $10^5$ & $0$  & $ (0.4000,0.4000  ,  0.4000  ,  0.4000    ,0.4000)$ & $147.0367$\\
 \hline
  &  \multicolumn{2}{c}{\textbf{\textbf{one-step survival}}}\\
  \hline 
\# MC & iterations & result & resnorm\\
 \hline \vspace{-3.5mm}\\ 
 $10^7$ & $34$ &$(0.5100,0.4766,0.4677,0.4887,0.5194)$ & $3.1451e-07$\\
 $10^6$ & $43$   &$(0.5071,0.4952,0.4761,0.4710,0.5165)$ & $ 4.3222e-05$\\
 $10^5$ & $84$ &$(0.4998,0.5097,0.4984,0.4319 ,0.4713)$ & $9.5003e-06 $\\
 \hline
  &  \multicolumn{2}{c}{\textbf{\textbf{one-step survival pathwise sensitivities}}}\\
  \hline
\# MC & iterations & result & resnorm\\
 \hline  \vspace{-3.5mm}\\ 
 $10^7$ & $34$ &$(0.5102,0.4753,0.4661,0.4883,0.5202)$ & $4.3762e-07$\\
 $10^6$ & $43$ &$(0.5071,0.4952,0.4761,0.4710,0.5165)$ & $4.3222e-05$\\
 $10^5$ & $84$ &$(0.4998,0.5096,0.4984,0.4305,0.4713)$ & $9.4110e-06$\\
 \hline
 \end{tabular}
   \caption{Results of calibrations using {\tt lsqnonlin} with initial-values ($0.4,0.4,0.4,0.4,0.4$) for the standard, one-step survival (both with finite differences) and one-step survival pathwise sensitivities Monte Carlo estimator. Depending on the number of Monte Carlo simulations (\# MC) the table shows the number of iterations taken (iterations), the returned solution (result) and the squared norm of the residual (resnorm). The true value was $(0.51, 0.48, 0.45, 0.49, 0.52)$.}
  \label{table:calibrationOSSPW}
 \end{table}

}

\begin{table}
\centering
 \begin{tabular}{|cccc|}
    \multicolumn{1}{c}{{\tt lsqnonlin:}} &   \multicolumn{2}{c}{\textbf{\textbf{one-step survival}}}\\
  \hline 
\# MC & iterations & result & resnorm\\
 \hline \vspace{-3.5mm}\\ 
 $10^7$ & $56$ &$(0.3000,0.3054,0.5270,0.3720,0.3240)$ & $4.5501
$\\
 $10^6$ & $56 $ &$ (0.3000   , 0.3072   , 0.5269,    0.3711   , 0.3194 )$ & $4.4726$\\
 $10^5$ & $59 $ &$( 0.3000  ,  0.3001   , 0.5265  ,  0.3880  ,  0.3981 )$ & $4.5162 $\\
 \hline
  &  \multicolumn{2}{c}{\textbf{\textbf{one-step survival pathwise sensitivities}}}\\
  \hline
\# MC & iterations & result & resnorm\\
 \hline  \vspace{-3.5mm}\\ 
 $10^7$ & $47$ &$(0.5101,0.4782,0.4618,0.4881,0.5204)$ & $5.4531e-07$\\
 $10^6$ & $35$ &$ (0.5071   , 0.4963  ,  0.4669  ,  0.4756   , 0.5167)$ & $4.3761e-05 $\\
 $10^5$ & $54$ &$ ( 0.4994  ,  0.5109 ,   0.4934 ,   0.4564  ,  0.4770)$ & $2.7085e-05 $\\
 \hline
 \end{tabular}
 
    \caption{Results of calibrations using {\tt lsqnonlin} with initial-values ($0.3,0.3,0.3,0.3,0.3$) for the one-step survival (with finite differences) and one-step survival pathwise sensitivities Monte Carlo estimator. Depending on the number of Monte Carlo samples (\# MC) the table shows the number of iterations taken (iterations), the returned solution (result) and the squared norm of the residual (resnorm). The true value was $(0.51, 0.48, 0.45, 0.49, 0.52)$.}
  \label{table:calibrationOSSPWradius}
 \end{table}

%
%

\section{Conclusions}
We adjusted the idea of pathwise sensitivities to the idea of the one-step survival Monte Carlo method suggested by Glasserman and Staum \cite{staum}. It followed, that we were able to calculate the pathwise sensitivities of options with discontinuous payoff, namely barrier options. In the numerical results, we saw that these derivatives behave stable and can be calculated efficiently related to finite differences, without evaluating a second respectively a third path. It followed, that there appears no problem in choosing a discretization shift size for balancing accuracy and stability, which would depend on the input parameters and the underlying Greek.

In the case study, we calibrated a CoCo-Bond, which we modelled with time-dependent barrier levels and saw that the new pathwise sensitivity estimator outperformed one-step survival (with finite differences) in terms of computation time. Furthermore, we determined a further advantage, since being without discretization error, the pathwise sensitivities improved the radius of convergence of the calibration.

The extension to non-constant parameters will work as follows: Using smaller time steps, it is assumable that the parameters are constant in these steps. Then, at the observation dates, our algorithm can be applied, with constant parameters and a smaller step width.

For a simplified presentation, we derived the algorithm within the classical Black-Scholes model. 
But, as we used a more general notation for the recursions of the paths and its derivatives, the extension to other models should be applicable as well, providing that the idea of one-step survival is feasible. The basic concepts on how to generalize one-step survival to the correlated multivariate case is given in \cite{alm} and we believe that the method of pathwise sensitivities can be applied as well. The authors divide the multivariate case into payoffs depending on the maximum or the minimum. We believe that pathwise sensitivities can be easily extended for the first case (maximum), one has to be careful with the second case (minimum), since the resulting estimator is Lipschitz-continuous but not necessarily differentiable.
We also believe, that the method can be adapted to more general models, e.g.\ models whose underlying depend upon the solution of stochastic differential equations, proceeding as follows: Modifying the approximation method, e.g.\ Euler-Maruyama or Milstein scheme, while sampling the approximation conditional on survival at specified approximation steps. Considering discretely observed barrier options this method leads to several usual performed approximation steps and one modified step at the observation date ensuring that the path survives, see e.g.\ \cite{gilesvibrato} for a similar combination idea. For continuously monitored barrier options, one could try to combine the one-step survival idea with the Brownian bridge interpolation in every step, see e.g.\ \cite{glasserman}, to smooth out the discontinuities of the barrier crossing probabilities. However, this would end in a loss of linearity, since many dependencies would be added. We also believe, that the method of pathwise sensitivities can be adapted to other kinds of options with path-dependent discontinuous payoff.

At last, it should be mentioned that other variance reduction methods, as control variates or antithetic sampling \cite{glasserman}, can be combined with the new algorithm.

\appendix
\section{Partial derivatives}
We present the partial derivatives of \eqref{pup2}, \eqref{sup2} with respect to $(\vartheta_1,\dots,\vartheta_4)$.\\
The partial derivatives of \eqref{pup2} with respect to $(\vartheta_1,\dots,\vartheta_4,s)$ are given by:
{\small
\begin{align}
&D^{(\vartheta_1,\dots,\vartheta_5,s)}(f(s,\vartheta))=\notag\\
&\left(\begin{array}{c} 0\\
 \sqrt{2}\, \mathrm{e}^{-\frac{{\left(\mathrm{log}\!\left(\frac{\vartheta_2}{s}\right) - \vartheta_5\, \left(\vartheta_3 - \frac{{\vartheta_4}^2}{2}\right)\right)}^2}{2\, {\vartheta_4}^2\, \vartheta_5}}/\left(2\, \vartheta_2\, \vartheta_4\, \sqrt{\vartheta_5}\, \sqrt{\pi}\right)\\
  -\sqrt{2}\, \sqrt{\vartheta_5}\, \mathrm{e}^{-\frac{{\left(\mathrm{log}\!\left(\frac{\vartheta_2}{s}\right) - \vartheta_5\, \left(\vartheta_3 - \frac{{\vartheta_4}^2}{2}\right)\right)}^2}{2\, {\vartheta_4}^2\, \vartheta_5}}/\left(2\, \vartheta_4\, \sqrt{\pi}\right)\\
    -\mathrm{e}^{-\frac{{\left(\mathrm{log}\!\left(\frac{\vartheta_2}{s}\right) - \vartheta_5\, \left(\vartheta_3 - \frac{{\vartheta_4}^2}{2}\right)\right)}^2}{2\, {\vartheta_4}^2\, \vartheta_5}}\, \left(\frac{\sqrt{2}\, \left(\vartheta_3 - \frac{{\vartheta_4}^2}{2}\right)}{\sqrt{\pi}2\, \vartheta_4\, \sqrt{\vartheta_5}} + \frac{\sqrt{2}\, \left(\mathrm{log}\!\left(\frac{\vartheta_2}{s}\right) - \vartheta_5\, \left(\vartheta_3 - \frac{{\vartheta_4}^2}{2}\right)\right)}{\sqrt{\pi}4\, \vartheta_4\, {\vartheta_5}^{\frac{3}{2}}}\right)\\
     -\sqrt{2}\, \mathrm{e}^{-\frac{{\left(\mathrm{log}\!\left(\frac{\vartheta_2}{s}\right) - \vartheta_5\, \left(\vartheta_3 - \frac{{\vartheta_4}^2}{2}\right)\right)}^2}{2\, {\vartheta_4}^2\, \vartheta_5}}/\left(2\, s\, \vartheta_4\, \sqrt{\vartheta_5}\, \sqrt{\pi})\right) \end{array}\right).
\label{Df}
\end{align}}
The partial derivatives of \eqref{sup2} with respect to $(\vartheta_1,\dots,\vartheta_4,s,\pi)$ are given by:
{\small
\begin{align}
D^{(\vartheta_1,\dots,\vartheta_5,s,\pi)}&(g\left(\pi,u,s,\vartheta\right))=\notag\\
&\left(\begin{array}{c} 0\\
 0\\
  s\, \vartheta_5\, \mathrm{e}^{\vartheta_5\, \left(\vartheta_3 - \frac{{\vartheta_4}^2}{2}\right) +  \vartheta_4\, \sqrt{\vartheta_5}\, \Phi^{-1}\!\left( \mathrm{\pi}\, u \right)}\\
    s\, \mathrm{e}^{\vartheta_5\, \left(\vartheta_3 - \frac{{\vartheta_4}^2}{2}\right) +  \vartheta_4\, \sqrt{\vartheta_5}\, \Phi^{-1}\!\left(\, \mathrm{\pi}\, u\right)}\, \left(\vartheta_3 - \frac{{\vartheta_4}^2}{2} + \frac{\, \vartheta_4\, \Phi^{-1}\!\left(\, \mathrm{\pi}\, u\right)}{2\, \sqrt{\vartheta_5}}\right)\\
     \sqrt{2}\, s\, u\, \vartheta_4\, \sqrt{\vartheta_5}\, \sqrt{\pi}\, \mathrm{e}^{\vartheta_5\, \left(\vartheta_3 - \frac{{\vartheta_4}^2}{2}\right) + \, \vartheta_4\, \sqrt{\vartheta_5}\, \Phi^{-1}\!\left(\, \mathrm{\pi}\, u \right)}\, \mathrm{e}^{{\Phi^{-1}\!\left( \mathrm{\pi}\, u \right)}^2/\sqrt{2}}\\ 
     \mathrm{e}^{\vartheta_5\, \left(\vartheta_3 - \frac{{\vartheta_4}^2}{2}\right) +  \vartheta_4\, \sqrt{\vartheta_5}\, \Phi^{-1}\!\left( \mathrm{\pi}\, u \right)} \end{array}\right).
\label{Dg}
\end{align}}

\bibliographystyle{siamplain}
\bibliography{references}

\end{document}